\documentclass[runningheads]{llncs}
\usepackage{graphicx}
\usepackage[sort]{cite}
\usepackage{enumerate}
\usepackage{booktabs}
\usepackage{tabu}
\usepackage[ruled,lined,boxed,linesnumbered]{algorithm2e}
\usepackage{amsmath,amssymb}
\usepackage{todonotes}
\usepackage{algorithmic}

\newtheorem{Reduction Rule}{Reduction Rule}

\usepackage{thmtools}
\usepackage{thm-restate}

\usepackage{hyperref}

\usepackage{cleveref}

\usepackage{microtype}
\usepackage{xspace}
\usepackage{color}

\newcommand{\sstpath}{{shortest $s$-$t$ path}\xspace}
\newcommand{\stpath}{{$s$-$t$ path}\xspace}
\newcommand{\stgraph}{{$s$-$t$ graph}\xspace}

\newcommand{\sstpaths}{{shortest $s$-$t$ paths}\xspace}
\newcommand{\stpaths}{{$s$-$t$ paths}\xspace}

\newcommand{\tsc}{{\sf tracking set condition}\xspace}
\newcommand{\tp}{{\sc Tracking Paths}\xspace}
\newcommand{\tpe}{{\sc Tracking Paths using Edges}\xspace}
\newcommand{\tsp}{{\sc Tracking Shortest Paths}\xspace}

\newcommand{\NP}{\text{\normalfont  NP}\xspace}

\newcommand{\APX}{{\sc APX}\xspace}

\newcommand{\Oh}{\mathcal{O}}

\newcommand{\defproblem}[3]{
  \vspace{1mm}
\noindent\fbox{
  \begin{minipage}{0.96\textwidth}
  \begin{tabular*}{\textwidth}{@{\extracolsep{\fill}}lr} #1 \\ \end{tabular*}
  {\bf{Input:}} #2  \\
  {\bf{Question:}} #3
  \end{minipage}
  }
  \vspace{1mm}
}

\begin{document}
\title{Polynomial Time Algorithms for Tracking Path Problems}
%
%

%
%

\author{Pratibha Choudhary}

\institute{Indian Institute of Technology Jodhpur, Jodhpur, India.\footnote{This work was done while the author was visiting The Institute of Mathematical Sciences, Chennai, India.}\\
\email{pratibhac247@gmail.com}
}

\authorrunning{P. Choudhary}

%
%
\maketitle              

\begin{abstract}
Given a graph $G$, and terminal vertices $s$ and $t$, the \tp problem asks to compute a minimum number of vertices to be marked as trackers, such that the sequence of trackers encountered in each \stpath is unique. \tp is \textsc{NP}-hard in both directed and undirected graphs in general. In this paper we give a collection of polynomial time algorithms for some restricted versions of \tp. 
We prove that \tp is polynomial time solvable for chordal graphs and tournament graphs. We prove that \tp is \textsc{NP}-hard in graphs with bounded maximum degree $\delta\geq 6$, and give a $2(\delta+1)$-approximate algorithm for the same. We also analyze the version of tracking \stpaths where paths are tracked using edges instead of vertices, and we give a polynomial time algorithm for the same. Finally, we show how to reconstruct an \stpath, given a sequence of trackers and a tracking set for the graph in consideration. 

\end{abstract}

\section{Introduction}
\label{sec:intro}

Tracking moving objects in networks has been studied extensively due to applications in surveillance and monitoring. Specific cases include secure system surveillance, habitat monitoring, vehicle tracking, and other similar scenarios. Object tracking in networks also finds applications in analyzing disease spreading patterns, information dissemination patterns on social media, and data packet flow in large networks like the world wide web. Tracking has been largely studied in the fields of machine learning, artificial intelligence, networking systems among other fields.

The problem of tracking paths in a network was first graphically modeled by Banik et al. in~\cite{ciac17}.
Let $G=(V,E)$ be an undirected graph without any self loops or parallel edges and suppose that $G$ has a unique entry vertex (source) $s$ and a unique exit vertex (destination) $t$. A simple path from $s$ to $t$ is called an \stpath.
The problem requires finding a set of vertices $T\subseteq V$, such that for any two distinct \stpaths, say $P_1$ and $P_2$, in $G$, the sequence of vertices in $T\cap V(P_1)$ as encountered in $P_1$ is different from the sequence of vertices in $T\cap V(P_2)$ as encountered in $P_2$. Here $T$ is called a \textit{tracking set} for the graph $G$, and the vertices in $T$ are referred to as \textit{trackers}.
Banik et al.~\cite{ciac17} proved that the problem of finding a minimum-cardinality tracking set to track {\em shortest} \stpaths (\tsp problem) is \NP-hard and \APX-hard. Later, the problem of tracking all \stpaths (\tp) in an undirected graph was studied in~\cite{tr-j},\cite{quadratic},\cite{ep-planar}. \tp is formally defined as follows.

\defproblem{\tp $(G,s,t)$}{An undirected  graph $G=(V,E)$ with terminal vertices $s$ and $t$.}
{Find a minimum cardinality tracking set  $T$ for $G$.}
\medskip

\tp was proven to be \textsc{NP}-complete in~\cite{tr-j}. Here, the authors studied the parameterized version of \tp, which asks if there exists a tracking set of size at most $k$, and showed it to be fixed-parameter tractable by giving a polynomial kernel. Specifically, it was proven that an instance of \tp can be reduced to an equivalent instance of size $\Oh(k^7)$ in polynomial time, where $k$ is the desired size of the tracking set. In~\cite{quadratic}, the authors improved this kernel to $\Oh(k^2)$, and gave a $\Oh(k)$ kernel for planar graphs. In~\cite{ep-planar}, Eppstein et al. proved that \tp is \textsc{NP}-complete for planar graphs and gave a $4$-approximation algorithm  for this setting. Here, the authors also proved that \tp can be solved in linear time for graphs of bounded clique width, when the clique decomposition is given in advance. 

\tsp was also studied in~\cite{caldam18} and \cite{guido-cubic}. In~\cite{caldam18}, Banik et al. studied \tsp and proved the problem to be fixed-parameter tractable. In~\cite{guido-cubic}, Bil{\`o} et al. prove that \tsp is \textsc{NP}-hard for cubic planar graphs in case of multiple source-destination pairs, and give an \textsc{FPT} algorithm parameterized by the number of vertices equidistant from the source $s$. 

In this paper we study \tp for chordal graphs, tournament graphs, and degree bounded graphs. A \textit{Chordal} graph is a graph in which each cycle of length greater than three has a chord (an edge between non-adjacent vertices of the cycle). A \textit{tournament} is a directed graph in which there exists a directed edge between each pair of vertices. So far all the work done on \tp has been focused on tracking \stpaths (or \sstpaths) using vertices. In this paper, we also study tracking \stpaths using edges. A natural question that follows is that of path reconstruction, which has already been studied for shortest paths~\cite{ciac17}. Here we give an algorithm for path reconstruction when all \stpaths are considered. 
Chordal graphs find applications in computational biology, computer vision and artificial intelligence~\cite{chordal-combio}, \cite{chordal-expert}, \cite{chordal-1}, \cite{chordal-2}. Tournament graphs are used in voting theory and social choice theory to graphically depict pairwise relationships between entities in a community~\cite{voting}, \cite{voting2}. Tournament graphs are particularly used to study the Condorcet voting model, where a preference is indicated between each pair of contestants~\cite{fisher}.

\noindent 
{\bf Our Results and Methods.} In this paper we give some polynomial time results for some variants of the \tp problem. We prove that \tp is polynomial time solvable for chordal graphs and tournaments. From~\cite{tr-j}, it is known that each cycle in the input graph needs a tracker. The key idea in proofs for chordal and tournament graphs is that if two \stpaths differ in only one vertex, than that vertex necessarily needs to be marked as a tracker. Next we prove that \tp is \textsc{NP}-hard for graphs with maximum degree $\delta$ ($\delta\geq 6$). We also give a $2(\delta+1)$-approximation algorithm for graphs with maximum degree $\delta$. Here the idea is to ensure that sufficient vertices are marked as trackers in each cycle. This derives from the fact that each cycle in a graph necessarily needs a tracker~\cite{tr-j}. In order to give a complete solution for tracking paths in a graph, we also give an algorithm that reconstructs the required \stpath given a sequence of trackers and a tracking set for the input graph. This uses the fact that by the definition of a tracking set, each maximal sequence of trackers in a tracking set should correspond to at most one \stpath in a graph. The reconstruction algorithm uses the disjoint path algorithms for undirected graphs~\cite{KAWARABAYASHI2012424} and tournaments~\cite{chudnovsky} to construct the required \stpath.

Towards the end of the paper we analyze the problem of tracking \stpaths in an undirected graph using edges rather than vertices. We prove that even while using edges, each cycle in the graph needs at least one edge to be marked as a tracker. Further, a minimum feedback edge set (set of edges whose removal makes a graph acyclic) is also a minimum tracking edge set.


\section{Definitions and Notations}
\label{sec:prelim}

Throughout the paper, while analyzing tracking paths using vertices in a graph, we assume graphs to be simple i.e. there are no self loops and multi-edges.
When considering tracking set for a graph $G=(V,E)$, we assume that the given graph is an \stgraph, i.e. the graph contains a unique source $s\in V$ and a unique destination $t\in V$ (both $s$ and $t$ are known), and we aim to find a tracking set that can distinguish between all simple paths between $s$ and $t$. Here $s$ and $t$ are also referred as the terminal vertices. In this paper, when we refer to tracking set, we mean tracking set for all \stpaths. If $a,b\in V$, then unless otherwise stated, $\{a,b\}$ represents the set of vertices $a$ and $b$, and $(a,b)$ represents an edge between $a$ and $b$. For a vertex $v\in V$, \textit{neighborhood} of $v$ is denoted by $N(v)=\{x \mid (x,v)\in E\}$. We use $deg(v)=|N(v)|$ to denote degree of vertex $v$. 
For a vertex $v\in V$ and a subgraph $G'$, $N_{G'}(v)=N(v)\cap V(G')$. For a subset of vertices $V'\subseteq V$ we use $N(V')$ to denote $\bigcup_{v\in V'} N(v)$. With slight abuse of notation we use $N(G')$ to denote $N(V(G'))$. For a graph $G$ and a set of vertices $S\subseteq V(G)$, $G-S$ denotes the subgraph induced by the vertex set $V(G)\setminus V(S)$. If $S$ is a singleton, we may use $G-x$ to denote $G-S$, where $S=\{x\}$. $[m]$ is used to denote the set of integers $\{1,\dots,m\}$. A \textit{chord} in a cycle is an edge between two vertices of the cycle, such that the edge itself not part of the cycle. In a directed cycle, a \textit{monotone} cycle is a cycle $C$ in which there exists a pair of vertices $a,b$ such that there exists two directed paths in $C$ from the vertex $a$ to the vertex $b$.

In a graph $G$, a feedback vertex set (FVS) is a set of vertices whose removal makes the graph acyclic. In an undirected graph $G$, feedback edge set (FES) is the set of edges whose removal makes the graph acyclic. An edge weighted undirected graph is an undirected graph with weights (real number values) assigned to each of its edges.
For a path $P$, $V(P)$ denotes the vertex set of path $P$ and $E(P)$ denotes the edge set of path $P$. For a subgraph (or graph) $G'$, $V(G')$ denotes the vertex set of $G'$, and $E(G')$ denotes the edge set of $G'$. Let $P_1$ be a path between vertices $a$ and $b$, and $P_2$ be a path between vertices $b$ and $c$, such that $V(P_1)\cap V(P_2)=\{b\}$. By $P_1 \cdot P_2$, we denote the path between $a$ and $c$, formed by concatenating paths $P_1$ and $P_2$ at $b$. Two paths $P_1$ and $P_2$ are said to be \textit{vertex disjoint} if their vertex sets do not intersect except possibly at the end points, i.e. $V(P_1)\cap V(P_2) \subseteq \{a,b\}$, where $a$ and $b$ are the starting and end points of the paths. By distance we mean length of the shortest path, i.e. the number of edges in that path. For a sequence of vertices $\pi$, by $V(\pi)$ we mean the set of vertices in the sequence $\pi$.
For a graph $G=(V,E)$, an FVS is a set of vertices $S\subseteq V$ such that $G\setminus S$ is a forest. If there exists a path $P$ such that $(a,b)$ is an edge that lies at one end point of $P$, then $P-(a,b)$ denotes the subpath of $P$ obtained after removing the edge $(a,b)$. Graphs which have maximum degree of vertices as three are known as \textit{cubic} graphs. By a \textit{bounded degree graph}, we mean a graph whose vertices have a maximum degree of $d$, where $d$ is some constant. 

\section{Preliminaries}

In this section, we give some basic claims which are necessary for the proofs of results in subsequent sections. Here the focus is on analysis of \tp in undirected graphs.
We start by first recalling a reduction rule from~\cite{tr-j} that ensures that each vertex and edge in the input graph participates in an \stpath.

\begin{Reduction Rule}~\cite{tr-j}
\label{red:stpath-undirected}
In graph $G$, if there exists a vertex or an edge that does not participate in any \stpath in $G$, then delete it.
\end{Reduction Rule}

It is known that Reduction Rule is safe and can be applied in quadratic time~\cite{tr-j} on undirected graphs. In rest of the paper, by \textit{reduced graph}, we mean a graph that is preprocessed using Reduction Rule~\ref{red:stpath-undirected}.
Next we recall the following lemma from~\cite{tr-j}, which is used to define some commonly used terms in this paper.

\begin{lemma}
\label{lemma:local-s-t}
In a reduced graph $G$, any induced subgraph $G'$ consisting of at least one edge, contains of a pair of vertices $u,v\in V(G')$ such that, \textsf{(a)} there exists a path in $G$ from $s$ to $u$, say $P_{su}$, and another path from $v$ to $t$, say $P_{vt}$, \textsf{(b)} $V(P_{su})\cap V(P_{vt})=\emptyset$, \textsf{(c)} $V(P_{su})\cap V(G')=\{u\}$ and $V(P_{vt})\cap V(G')=\{v\}$.
\end{lemma}

With respect to Lemma~\ref{lemma:local-s-t}, we refer to the vertex $u$ as a \textit{local source} and vertex $v$ as a \textit{local destination} for the subgraph $G'$. Note that a subgraph may have multiple local source-destination pairs.

Now we state the \textit{tracking set condition}, which is useful for validation of a tracking set~\cite{tr-j}.

\vspace{1mm}
\noindent\fbox{
  \begin{minipage}{0.96\textwidth}
 
\textbf{Tracking Set Condition:}

For a graph $G=(V,E)$, with terminal vertices $s,t\in V$, a set of vertices $T\subseteq V$, is said to satisfy the \tsc if there does not exist a pair of vertices $u,v\in V$, such that the following holds:

\begin{itemize}

\item there exist two distinct paths, say $P_1$ and $P_2$, between $u$ and $v$ in $(G\setminus (T\cup\{s,t\})) \cup \{u,v\}$, and
\item there exists a path from $s$ to $u$, say $P_{su}$, and a path from $v$ to $t$, say $P_{vt}$, in $(G\setminus (V(P_1)\cup V(P_2))) \cup \{u,v\}$, and $V(P_{su}) \cap V(P_{vt}) = \emptyset$, i.e. $P_{su}$ and $P_{vt}$ are mutually vertex disjoint, and also vertex disjoint from $P_1$ and $P_2$.

\end{itemize}
 
  \end{minipage}
  }
  
\vspace{1mm}

It is known that for a reduced graph $G$, a set of vertices $T\subseteq V(G)$ is a tracking set if and only if $T$ satisfies the \textit{tracking set condition}. We use this fact, to prove the following lemma.

\begin{lemma}
\label{lemma:not-trs-cycle}
In a graph $G$, if $T\subseteq V(G)$ is not a tracking set for $G$, then there exists two \stpaths with the same sequence of trackers, and they form a cycle $C$ in $G$, such that $C$ has a local source $a$ and a local destination $b$, and $T\cap (V(C)\setminus\{a,b\})=\emptyset$.
\end{lemma}
\begin{proof}
Let $G$ be a graph, such that $T\subseteq V(G)$ is not a tracking set for $G$. Due to tracking set condition, it is known that in such a case, there exists two distinct vertices $u,v$ along with two distinct paths $P_1,P_2$ between $u$ and $v$, such that there are no trackers on $P_1$ and $P_2$ except possibly at $u$ and $v$. Further, there exists a path $P_{su}$ from $s$ to $u$ and a path $P_{vt}$ from $v$ to $t$, and such that these $P_{su}$ and $P_{vt}$ are vertex disjoint, and they intersect with $P_1$ and $P_2$ only at $u$ and $v$. Let $G'$ be the graph induced by $V(P_1)\cup V(P_2)$. Observe that $u$ and $v$ form a local source-destination pair for $G'$. Note that no vertex, other than possible $u$ and $v$, in $G'$ is a tracker.
Starting from $u$, let $a$ be the last vertex in $G'$ until which paths $P_1$ and $P_2$ have the same sequence of vertices. Let $b\in V(P_1)$ be the first vertex in $P_1$ after $a$, such that $b\in V(P_1)\cap V(P_2)$. We use $P_{ab_1}$ to denote the subpath of $P_1$ lying between vertices $a$ and $b$, and $P_{ab_2}$ to denote the subpath of $P_2$ lying between the vertices $a$ and $b$. Observe that paths $P_{ab_1}$ and $P_{ab_2}$ are vertex disjoint (except for vertices $a$ and $b$) and thus form a cycle, say $C$. Further there exists a subpath of $P_2$ between $b$ and $v$, that intersects $C$ only at $b$. Since $P_1$ and $P_2$ share the same vertex sequence from $u$ to $a$, $a$ is a local source for $C$. Also, by construction, $b$ is a local destination for $C$. Note that it is possible that $u=a$ and/or $b=v$. However, $V(C)\setminus\{a,b\}$ does not contain any trackers. Hence the lemma holds.
\qed
\end{proof}

\section{Tracking Paths in Chordal Graphs and Tournaments}

In this section, we give polynomial time algorithms for solving \tp for chordal graphs and tournaments.

\subsection{Chordal Graphs}
Here we give a polynomial time algorithm to find a tracking set for undirected chordal graphs. 
Recall that chordal graphs are those graphs in which each cycle of length greater than three has a chord. Many problems that are known to be \textsc{NP}-hard on general graphs are polynomial time solvable for chordal graphs e.g. chromatic number, feedback vertex set, independent set~\cite{golumbic}.

In undirected graphs, a tracking set is also a feedback vertex set~\cite{tr-j}. However, a tracking set can be arbitrarily larger in size compared to a feedback vertex set. This holds true for chordal graphs as well. See Figure~\ref{fig:chordal-fvs-not-trs}. 
\begin{figure}[ht]
\centering
\includegraphics[scale=0.7]{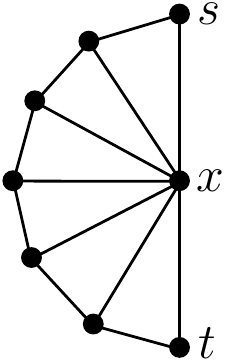} 
\caption{Depiction of a chordal graph on $n$ vertices with an optimum tracking set ($V(G)\setminus\{s,t\}$) of size $n-2$ and an FVS (vertex $x$) of size $1$} 
\label{fig:chordal-fvs-not-trs}
\end{figure}

\begin{algorithm}[ht]
\caption{Finding Tracking Set for a Chordal Graph.}
\label{alg:chordal}


\KwIn{Chordal graph $G=(V,E)$ and vertices $s,t\in V$.}
\KwOut{Tracking Set $T\subseteq V$ for $G$.}

\SetAlgoLined

\BlankLine
 Initialize $T=\emptyset$;
Apply Reduction Rule~\ref{red:stpath-undirected}\;
\BlankLine
\ForEach{$e=(a,b) \in E$}{  
  \ForEach{$x\in (N(a)\cap N(b))\setminus T$}{
  	\If{$\exists$ an \stpath $P$ in $G-x$ such that $e\in E(P)$}{
  		$T=T\cup\{x\}$\;
  	}
  }
  }
    Return $T$\;

\end{algorithm}

Algorithm~\ref{alg:chordal} gives a procedure to compute a minimum tracking set for a chordal graph $G$. We prove its correctness in the following lemma.

\begin{lemma}
\label{lemma:chordal-algo-correctness}
Algorithm~\ref{alg:chordal} gives an optimum tracking set for a chordal graph.
\end{lemma}
\begin{proof}
Algorithm~\ref{alg:chordal} starts by ensuring that each vertex and edge in the input graph $G$ participates in an \stpath of $G$. Next for each edge $e=(a,b)\in E$, if there exists a vertex $x\in (N(a)\cap N(b))\setminus T$, we check if there exists an \stpath in $G-x$ that contains the edge $e$. Let $P$ such a path in $G-x$. Now consider the path $P'$ that can be obtained by replacing the edge $e$ in $P$ by the path $(a,x)\cdot(b,x)$ along with the vertex $x$. Observe that the vertex sets of $P$ and $P'$ differ only in vertex $x$. Hence, $x$ necessarily belongs to a tracking set for $G$. 

Now we prove that Algorithm~\ref{alg:chordal} indeed returns an optimal tracking set $T$ for $G$. Suppose not. Then $T$ is not a tracking set for $G$. Due to Lemma~\ref{lemma:not-trs-cycle}, there exists two \stpaths, say $P_1,P_2$ and they form a cycle $C$ in $G$, such that $C$ has a local source $u$ and a local destination $v$, and $V(C)\setminus\{u,v\}$ does not contain any trackers. See Figure~\ref{fig:paths-cycle}.

\begin{figure}[ht]
\centering
\includegraphics[scale=0.6]{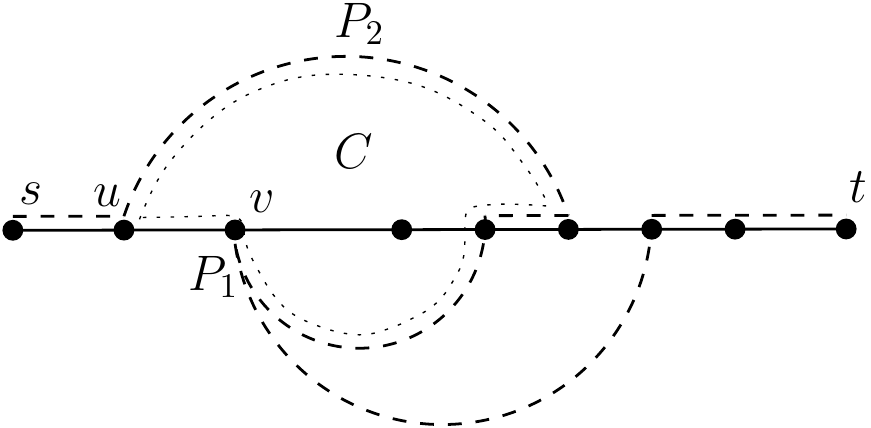} 
\caption{Indistinguishable \stpaths in a graph form a cycle (marked in dotted lines)} 
\label{fig:paths-cycle}
\end{figure}

Path $P_1$ is marked in solid lines, while path $P_2$ is marked in dashed lines. Since $P_1$ and $P_2$ contain the same sequence of trackers, no vertex in $V(C)\setminus\{u,v\}$ can be a tracker. Since we consider graphs without any parallel edges, there exists at least one vertex in $V(C)\setminus\{u,v\}$. 

First we consider the case where $C$ is a triangle. Due to Algorithm~\ref{alg:chordal}, the vertex in $V(C)\setminus\{u,v\}$ would necessarily have been marked as a tracker. This contradicts the assumption that no vertex in $V(C)\setminus \{u,v\}$ is marked as a tracker. 

Next we consider the case when $C$ is not a triangle, i.e. $C$ is a cycle containing four or more vertices. Since $G$ is a chordal graph, $C$ necessarily contains a chord. We consider the following two cases based on whether a chord is incident on the vertex $u$ or not. Let $w$ and $x$ be two vertices such that $w,x\in N(u)\cap V(C)$ and $(w,x)\in E$. Without loss of generality, let $x\in V(P_1)$. Observe that the edge $(u,x)$ in path $P_1$ can be replaced by the concatenated path $(u,w)\cdot(w,x)$, to obtain a new path that differs from $P_1$ only at the vertex $w$. Hence $w$ must have been marked as a tracker in Algorithm~\ref{alg:chordal}. Next we consider the case when a chord in $C$ is incident on $u$. Let $a\in N(u)\cap V(C)$ and $b\in N(a)\cap V(C)$, such that $a\neq b\neq u$, and $(b,u)\in E$. Observe that there exists an \stpath containing the edge $(b,u)$, in which $(b,u)$ can be replaced with the path $(b,a)\cdot(a,u)$, to obtain a new path that differs only at the vertex $a$. Due to Algorithm~\ref{alg:chordal}, vertex $a$ must have been marked as a tracker. Both the above cases contradict the assumption that no vertex in $V(C)\setminus\{u,v\}$ is a tracker. Hence Algorithm~\ref{alg:chordal} gives an optimum tracking set for a chordal graph.
\qed
\end{proof}

Next we prove that Algorithm~\ref{alg:chordal} runs in polynomial time.

\begin{lemma}
\label{lemma:chordal-algo-time}
Algorithm~\ref{alg:chordal} runs in time $\Oh(m.n^3)$.
\end{lemma}
\begin{proof}
Let $G=(V,E)$ be the input graph, $|V|=n$ and $|E|=m$. Due to~\cite{tr-j}, it is known that it takes $\Oh(n^2)$ (quadratic) time to apply Reduction Rule~\ref{red:stpath-undirected}. Next for each edge $e=(a,b)\in E$, we consider the set of vertices that are adjacent to both end points $a,b$ of the edge $e$. This takes $\Oh(m.n)$ time. Now for each vertex $x$ that is adjacent to both $a$ and $b$, we check if $e$ participates in some \stpath in the graph $G-x$. Removal of vertex $x$ from $G$ takes constant time. In order to check if $e$ participates in some \stpath in $G-x$, we check if there exists a path between $s$ and $a$, say $P_1$, and a path between $b$ and $t$, say $P_2$, such that $V(P_1)\cap V(P_2)=\emptyset$. This can be done using the algorithm for finding vertex disjoint paths shown in~\cite{KAWARABAYASHI2012424} in $\Oh(n^2)$ time. Thus the total time taken to run Algorithm~\ref{alg:chordal} is $\Oh(n^2)+\Oh(m.n.n^2)$, i.e. $\Oh(m.n^3)$.
\qed
\end{proof}

From Lemma~\ref{lemma:chordal-algo-correctness} and Lemma~\ref{lemma:chordal-algo-time}, we have the following theorem.

\begin{theorem}
\label{theorem:chordal}
\tp can be solved in polynomial time in chordal graphs.
\end{theorem}

\subsection{Tournaments}
Here we give a polynomial time algorithm to find a tracking set for tournament graphs. Recall that tournaments are directed graphs that have a directed edge between each pair of vertices in the graph. A lot of problems including feedback vertex set and feedback arc set are known to be \textsc{NP}-hard in tournaments. From~\cite{tr1-j} it is known that \tp is \textsc{NP}-hard for directed acyclic graphs. Since a directed acyclic graph is also a directed graph, this implies that \tp is \textsc{NP}-hard for directed graphs as well. However, as we prove now, \tp is in \textsc{P} for tournament graphs.

We start by first applying Reduction Rule~\ref{red:stpath-undirected}, to ensure that each vertex and edge in the input graph participates in an \stpath. Lemma~\ref{lemma:red-stpath-tournament} in Appendix proves that Reduction Rule~\ref{red:stpath-undirected} is safe and can be applied on tournament graphs in polynomial time.

\begin{lemma}
\label{lemma:red-stpath-tournament}
Reduction Rule~\ref{red:stpath-undirected} is safe and can be applied in polynomial time in tournament graphs.
\end{lemma}
\begin{proof}
Consider a tournament graph $G=(V,E)$.
If a vertex or an edge in $G$ does not participate in any \stpath, then it can not contribute to tracking any paths, nor does it need to be considered while ensuring that all \stpaths have unique sequence of trackers. Hence Reduction Rule~\ref{red:stpath-undirected} is safe for tournament graphs.

In order to apply Reduction Rule~\ref{red:stpath-undirected}, for each edge $e=(a,b)\in E$, check if there exists a path from $s$ to $a$, say $P_1$, and a path from $b$ to $t$, say $P_2$, such that $V(P_1)\cap V(P_2)=\emptyset$, using the algorithm for finding vertex disjoint paths in tournament given in~\cite{chudnovsky} in $\Oh(n^2)$ time. If such paths do not exist, then delete the edge $e$. After repeating the process for all edges in $G$, we delete the isolated vertices. Thus the Reduction Rule~\ref{red:stpath-undirected} can be applied in $\Oh(m.n^2)$ time in a tournament graph.
\qed
\end{proof}

\begin{algorithm}[ht]
\caption{Finding Tracking Set for a Tournament Graph.}
\label{alg:tournament}


\KwIn{Tournament graph $G=(V,E)$ and vertices $s,t\in V$.}
\KwOut{Tracking Set $T\subseteq V$ for $G$.}

\SetAlgoLined

\BlankLine
 Initialize $T=\emptyset$\;
Apply Reduction Rule~\ref{red:stpath-undirected}\;
\BlankLine
\ForEach{$e=(a,b) \in E$}{  
  \ForEach{$x\in (N^+(a)\cap N^-(b))\setminus T$}{
  	\If{$\exists$ an \stpath $P$ in $G-x$ such that $e\in E(P)$}{
  		$T=T\cup\{x\}$\;
  	}
  }
  }
    Return $T$\;

\end{algorithm}

Algorithm~\ref{alg:tournament} gives a procedure to compute a minimum tracking set for a tournament graph $G$. Next two lemmas prove the correctness and running time of the algorithm.

\begin{lemma}
\label{lemma:tournament-algo-correctness}
Algorithm~\ref{alg:tournament} gives an optimum tracking set for a tournament graph.
\end{lemma}
\begin{proof}
Algorithm~\ref{alg:tournament} starts by ensuring that each vertex and edge in the input graph $G=(V,E)$ participates in an \stpath of $G$. Next for each edge $e=(a,b)\in E$, if there exists a vertex $x\in (N^+(a)\cap N^-(b))\setminus T$, we check if there exists an \stpath in $G-x$ that contains the edge $e$. Let $P$ such a path in $G-x$. Now consider the path $P'$ that can be obtained by replacing the edge $e$ in $P$ by edges $(a,x)$ and $(b,x)$ along with the vertex $x$. Observe that the vertex sets of $P$ and $P'$ differ only in vertex $x$. Hence, $x$ necessarily belongs to a tracking set for $G$. 

Now we prove that Algorithm~\ref{alg:tournament} indeed returns a tracking set for $G$. Suppose not. Then there exists two \stpaths, say $P_1$ and $P_2$, that contain the same sequence of trackers in $G$. Consider the graph $G'$ induced by $V(P_1)\cup V(P_2)$. Starting from $s$, let $u$ be the last vertex until which $P_1$ and $P_2$ contain the same sequence of vertices. Let $v\in V(P_1)$ be the first vertex on $P_1$ after $u$, such that $v\in V(P_2)$. Let $P_{uv_1}$ be the subpath of $P_1$ lying between the vertices $u$ and $v$, and $P_{uv_2}$ be the subpath of $P_2$ lying between the vertices $u$ and $v$. Observe that $P_{uv_1}$ and $P_{uv_2}$ are vertex disjoint except for their endpoints $u$ and $v$, hence they form a monotone cycle, say $C$. See Figure~\ref{fig:paths-cycle-directed}. Further, there exists a subpath of $P_2$ from $v$ to $t$ that intersects with $C$ only at $v$.

\begin{figure}
\begin{minipage}[h]{0.45\textwidth}
\centering
  {\includegraphics[scale=0.8]{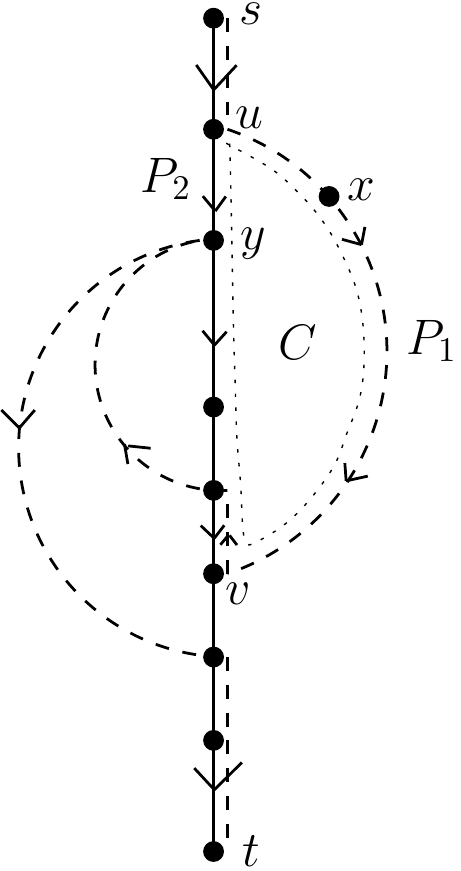}}  
     \caption{Indistinguishable \stpaths in a graph form a cycle (marked in dotted lines)}
     \label{fig:paths-cycle-directed}
\end{minipage}\hspace{10mm}
\begin{minipage}[h]{0.45\textwidth}
\centering
  {\includegraphics[scale=0.75]{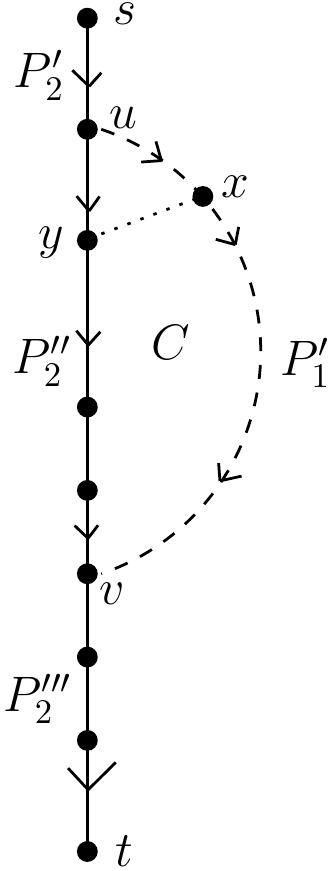}} 
  \vspace{4mm}
    \caption{Monotone cycle}
    \label{fig:monotone-cycle}
\end{minipage}
\end{figure} 

Path $P_1$ is marked in dashed lines, while path $P_2$ is marked in solid lines. Observe the monotone cycle $C$ formed due to paths $P_1$ and $P_2$ with two directed paths from the vertex $u$ to vertex $v$ in $C$. 
Since we consider graphs without any parallel edges, there exists at least one vertex in $V(C)\setminus\{u,v\}$. 

Let $x,y$ be the two out-neighbors of vertex $u$ such that $x,y\in V(C)$ and $x\in V(P_1)$ and $y\in V(P_2)$. Since $G$ is a tournament, there exists an edge between the vertices $x$ and $y$. See Figure~\ref{fig:monotone-cycle}. We use $P_2'$ to denote the subpath of $P_2$ lying between vertices $s$ and $u$, $P_2''$ to denote the subpath of $P_2$ between vertices $u$ and $v$, and $P_2'''$ to denote the subpath of $P_2$ between vertices $v$ and $t$. While $P_1'$ denotes the subpath of $P_1$ between vertices $u$ and $v$. First we consider the case when $(x,y)\in E$. Observe that paths $P_2$ and $P_2'\cdot (u,x) \cdot (x,y)\cdot P_2'-(u,y)$ are two \stpaths that differ only in vertex $x$. Hence $x$ must have been marked as tracker by Algorithm~\ref{alg:tournament}. Next consider the case when $(y,x)\in E$. Now paths $P_2'\cdot P_1'\cdot P_2'''$ and $P_2'\cdot (u,y)\cdot(y,x)\cdot P_1'-(u,x)\cdot P_2'''$ are two \stpaths that differ only at vertex $y$. Hence $y$ must have been marked as a tracker by Algorithm~\ref{alg:tournament}. This contradicts the assumption that no vertex other than $u,v$ in $C$ is marked as a tracker.
Hence Algorithm~\ref{alg:tournament} gives an optimum tracking set for a tournament graph.
\qed
\end{proof}

\begin{lemma}
\label{lemma:tournament-algo-time}
Algorithm~\ref{alg:tournament} runs in time $\Oh(m.n^3)$.
\end{lemma}
\begin{proof}
Let $G=(V,E)$ be the input graph, $|V|=n$ and $|E|=m$. From~\cite{tr-j}, it is known that it takes $\Oh(n^2)$ (quadratic) time to apply Reduction Rule~\ref{red:stpath-undirected}. Next for each edge $e=(a,b)\in E$, we consider the set of vertices that are adjacent to both end points $a,b$ of the edge $e$. This takes $\Oh(m.n)$ time. Now for each vertex $x$ that is adjacent to both $a$ and $b$, we check if $e$ participates in some \stpath in the graph $G-x$. Removal of vertex $x$ from $G$ takes constant time. In order to check if $e$ participates in some \stpath in $G-x$, we check if there exists a path between $s$ and $a$, say $P_1$, and a path between $b$ and $t$, say $P_2$, such that $V(P_1)\cap V(P_2)=\emptyset$. This can be done using the algorithm for finding vertex disjoint paths in tournaments shown in~\cite{chudnovsky} in $\Oh(n^2)$ time. Thus the total time taken to run Algorithm~\ref{alg:tournament} is $\Oh(n^2)+\Oh(m.n.n^2)$, i.e. $\Oh(m.n^3)$.
\qed
\end{proof}

From Lemma~\ref{lemma:tournament-algo-correctness} and Lemma~\ref{lemma:tournament-algo-time}, we have the following theorem.

\begin{theorem}
\label{theorem:tournament}
\tp can be solved in polynomial time in tournament graphs.
\end{theorem}

\section{Bounded-Degree Graphs}

In this section, we give an approximation algorithm for \tp. We show that given an undirected graph $G$, there exists a polynomial time algorithm that gives a tracking set for $G$, of the size $2(\delta+1)\cdot OPT$, where $OPT$ is the size of an optimum tracking set for $G$ and $\delta$ is the maximum degree of graph $G$. Approximation algorithms have been studied for restricted versions of \tsp and \tp. Banik et al. gave a $2$-approximate algorithm for \tsp in planar graphs in~\cite{ciac17}. Eppstein et al. gave a $4$-approximate algorithm for \tp in planar graphs in~\cite{ep-planar}. Bil{\`o} et al. gave a $\tilde{O}(\sqrt{n})$-approximate algorithm for \tsp in case of multiple source-destination pairs in~\cite{guido-cubic}.
Next we show that \tp for bounded degree graphs is polynomial time reducible from \textsc{Vertex Cover} for bounded degree graphs. 

\begin{lemma}
\label{lemma:bounded-degree-nphard}
Given an undirected graph $G$ with maximum degree $d$, there exists an $s$-$t$ graph $G'$ with maximum degree $2d$, such that $G$ has a vertex cover of size $k$ if and only if $G'$ has a tracking set for all \stpaths, of size  $k+|E|^2+3|E|-2$.
\end{lemma}
\begin{proof}
Let $G$ be and undirected graph with maximum degree $d$. For reference, let $G$ be the graph in Figure~\ref{fig:vc}. 

\begin{figure}[h]
\begin{minipage}[t]{0.25\textwidth}
\centering
  {\includegraphics[scale=0.5]{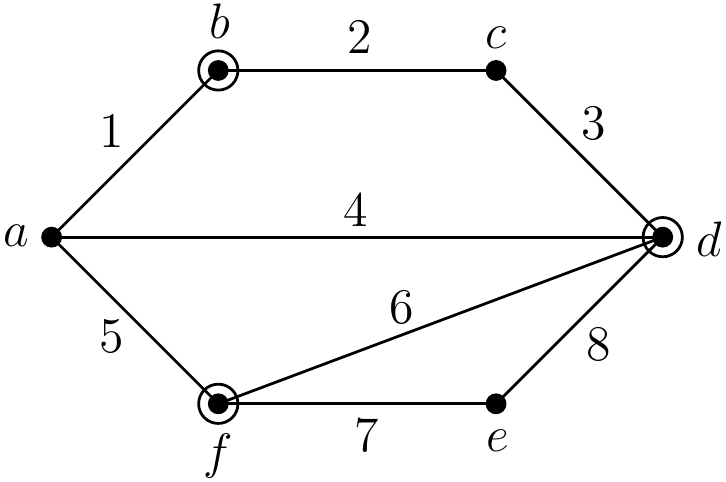}}  
     \caption{Depiction of an undirected graph $G$ with maximum degree $d$}
     \label{fig:vc}
\end{minipage}\hspace{8mm}
\begin{minipage}[h]{0.75\textwidth}
\centering
  {\includegraphics[scale=0.75]{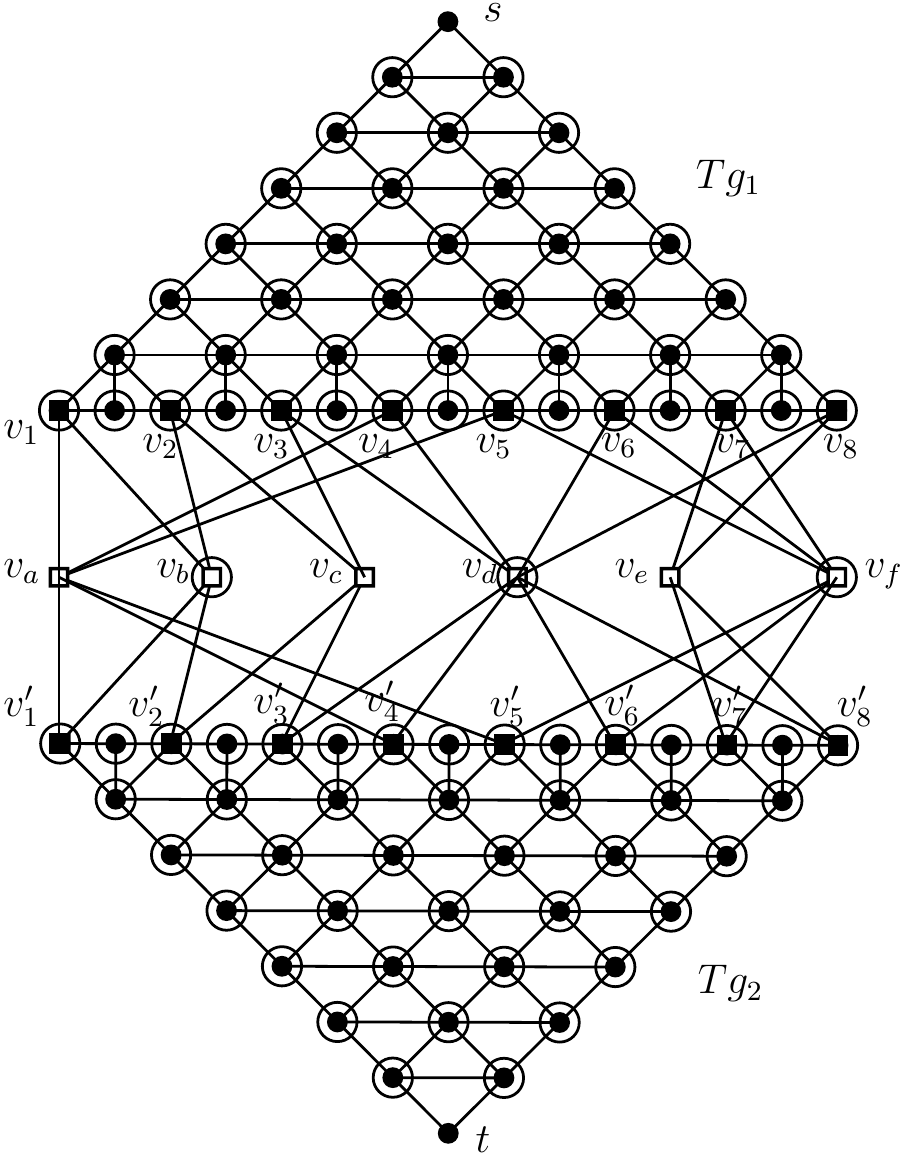}} 
  \vspace{4mm}
    \caption{Depiction of graph $G'$ mentioned in Lemma~\ref{lemma:bounded-degree-nphard}}
    \label{fig:bounded-deg-nphard}
\end{minipage}
\end{figure}

We create the graph $G'$ as follows. For each vertex $a\in V(G)$, we introduce a vertex $v_a$ in $V(G')$, and we call these set of vertices in $G'$ as $V_v$. For each edge $i\in E(G)$, we introduce two vertices $v_i, v_i'$ in $E(G')$, and we call the set of vertices $v_i$ as $V_e$, and the set of vertices $v_i'$ as $V_e'$. 
The adjacencies between $V_v$ and $V_e,V_e'$ are introduced as follows. If an edge $i$ is incident on vertices $a,b$ in $G$, then we add edges between the corresponding vertices $v_i,v_i'\in V_e,V_e'$ and the vertices $v_a,v_b\in V_v$ in $G'$. Next, we add the source and destination vertices $s$ and $t$ in $G'$. We then create a triangular grid $Tg_1$ between $s$ and the vertices in $V_e$, and another triangular grid between the vertices in $V_e'$ and $t$. See Figure~\ref{fig:bounded-deg-nphard}. The vertices of $V_v$ are marked with blank boxes, while the ones from $V_e\cup V_e'$ are marked with solid boxes. The circled vertices form a tracking set.
Observe that the maximum degree of vertices in $Tg_1$ and $Tg_2$, including the vertices in $V_e\cup V_e'$, is $6$. The maximum degree of vertices in $V_v$ is at most $2d$.

Now we prove that there exists a vertex cover of size $k$ in $G$ if and only if there exists a tracking set in $G'$ of size $k+|E|^2+3|E|-2$.
First consider the case when $G$ has a vertex cover $V_c$ of size $k$. We now prove that there exists a tracking set of size $k+|E|^2+3|E|-2$ in $G'$. We mark the vertices in $G'$ corresponding to $V_c$ as trackers. In addition we mark all the vertices in $Tg_1$ and $Tg_2$ (except $s$ and $t$) as trackers. Now the size of tracking set $T$ in $G'$ is $k+|E|^2+3|E|-2$. We claim that $T$ is a valid tracking set for $G'$. Suppose not. Then there exists two distinct \stpaths, say $P_1,P_2$ in $G'$, such that the sequence of trackers in $P_1$ is same as that in $P_2$. Observe that two distinct subpaths (subpaths of some \stpaths) contained in $Tg_1$ ($Tg_2$) cannot have the same sequence of trackers from $Tg_1-\{s\}$ ($Tg_2-\{t\}$). Since all vertices in $Tg_1,Tg_2$ are marked as trackers, this implies that $P_1,P_2$ contain the same sequence of vertices from $Tg_1$ and $Tg_2$, and they necessarily differ in vertices from $V_v$. Let $x,y\in V_v$ be the vertices that distinguish $P_1$ and $P_2$, and $x\in V(P_1)$ and $y\in V(P_2)$. Since $P_1$ and $P_2$ cant differ in their vertex set from $Tg_1$ and $Tg_2$, the vertex preceding $x,y$ has to be common in both $P_1,P_2$. Without loss of generality, we assume that the $z$ is the vertex preceding $x,y$, and $z\in V(Tg_1)$. This implies that $z\in V_e$. Note that $z$ corresponds to an edge in $G$. Since we marked the vertices corresponding to $V_c$ as trackers in $G'$, at least one of the neighbors of $z$ in $V_v$ is necessarily a tracker. Thus either $x$ or $y$ is necessarily a tracker. This contradicts the assumption that $P_1$ and $P_2$ have the same sequence of trackers.

Now we consider the case when $G'$ has a tracking set $T$ of size $k+|E|^2+3|E|-2$. We claim that there exists a vertex cover of size $k$ in $G$. Suppose not. Consider the triangular grid subgraphs $Tg_1$ and $Tg_2$. Observe that for each edge $(a,b)$ in $Tg_1$, there exists a vertex $c\in N(a)\cap N(b)$, and there exists an \stpath, say $P_1$, that passed through $(a,b)$ in $G-c$, such that we can replace edge $(a,b)$ in $P_1$ by edges $(a,c)$,$(c,b)$ to form another \stpath, say $P_2$. Observe that $P_1$ and $P_2$ differ in only one vertex i.e. $c$. Hence $c$ is necessarily a tracker. The same holds true for each edge in $Tg_2$. Thus all vertices in $V(Tg_1)\cup V(Tg_2) \setminus \{s,t\}$ are necessarily trackers and hence belong to $T$. Since $|V(Tg_1)\cup V(Tg_2) \setminus \{s,t\}|=|E|^2+3|E|-2$, the remaining $k$ trackers in $T$ are vertices from $V_v$. Let $V_t$ be the set of vertices in $V_v$ that have been marked as trackers, i.e. $V_t=V_v\cap T$. Note that $|V_t|=k$. We denote the set of vertices in $G$ that correspond to vertices in $V_t$ as $V_c$. We claim that $V_c$ forms a vertex cover for $G$. Suppose not. Then there exists an edge, say $(a,b)$ in $G$, such none of its end points $a,b$ belong to $V_c$. This implies that the vertices in $V_v$ that correspond to $a$ and $b$, say $v_a,v_B$, are not trackers in $G'$. Due to the construction of $G'$, there exists a pair of vertices $v_i\in V_e$ and $v_i'\in V_e'$ ($v_i,v_i'$ correspond to the edge $(a,b)$ in $G$) such that $v_a$ and $v_b$ are adjacent to both $v_i$ and $v_i'$.

Observe that for each pair of vertices $v_i,v_i'$, where $v_i\in V_e$ and $v_i'\in V_e'$, there exists two vertices in $V_v$ (the vertices in $V(G)$ that correspond to the endpoints of the edge $i$ in $G$) that are adjacent to both $v_i$ and $v_i'$. Thus for each pair of vertices $v_i,v_i'$, there exists two paths between them passing through two distinct vertices in $V_v$. Further, there exists a path from $s$ to $v_i$ that is completely contained in $Tg_1$, and there exists a path from $v_i'$ to $t$ that is completely contained in $Tg_2$. Thus at least one of the vertices from $V_v$ that are adjacent to $v_i,v_i'$, must necessarily be a tracker. This contradicts the fact that neither $v_a$ nor $v_b$ is a tracker in $G'$. This completes the proof.
\qed
\end{proof}

Since \textsc{Vertex Cover} is known to be \textsc{NP}-hard for graphs with maximum degree $d$ ($d\geq 3$)~\cite{vc-cubic}, due to Lemma~\ref{lemma:bounded-degree-nphard} we have the following corollary.

\begin{corollary}
\tp is \textsc{NP}-hard for graphs with maximum degree $\delta\geq6$.
\end{corollary}

\begin{algorithm}[ht]
\caption{Finding a $2(\delta+1)$-approximate tracking set for undirected graphs with maximum degree $\delta$.}
\label{alg:approximate}


\KwIn{Undirected graph $G=(V,E)$ such that $deg(x)\leq\delta$, $\forall x\in V$, and vertices $s,t\in V$.}
\KwOut{Tracking Set $T\subseteq V$ for $G$.}

\SetAlgoLined

\BlankLine
Apply Reduction Rule~\ref{red:stpath-undirected}\;
Find a $2$-approximate feedback vertex set $S$ for $G$\;
Set $T=S$\;
\ForEach{$v\in S$}{ 
  \ForEach{$x\in N(v)$}
  	{$T=T\cup\{x\}$\;
  }
  }
    Return $T$\;

\end{algorithm}

Algorithm~\ref{alg:approximate} gives a procedure to find a $2(\delta+1)$-approximate tracking set for undirected graphs with maximum degree $\delta$. We prove its correctness in the following lemma.

\begin{lemma}
\label{lemma:approximate-algo-correctness}
Algorithm~\ref{alg:approximate} gives a $2(\delta+1)$-approximate tracking set for an undirected graph.
\end{lemma}
\begin{proof}
Algorithm~\ref{alg:approximate} starts by ensuring that each vertex and edge in the input graph $G$ participates in an \stpath of $G$ by applying Reduction Rule~\ref{red:stpath-undirected}. 

Next we claim that Algorithm~\ref{alg:approximate} indeed returns an optimal tracking set $T$ for $G$. Suppose not. Then $T$ is not a tracking set for $G$. Due to Lemma~\ref{lemma:not-trs-cycle}, there exists a cycle $C$ in $G$ with a local source $u$ and a local destination $v$, such that $V(C)\setminus\{u,v\}$ does not contain any trackers. See Figure~\ref{fig:paths-cycle}.

Path $P_1$ is marked in solid lines, while path $P_2$ is marked in dashed lines. Observe the cycle $C$ formed due to paths $P_1$ and $P_2$. Since $P_1$ and $P_2$ contain the same sequence of trackers, no vertex in $V(C)\setminus\{u,v\}$ can be a tracker. Since we consider graphs without any parallel edges, there exists at least one vertex in $V(C)\setminus\{u,v\}$. 
Note that Algorithm~\ref{alg:approximate} includes a $2$-approximate feedback vertex set, $S$, for $G$ in $T$. Thus at least one vertex from $C$ belongs to $T$. Note that it is possible that vertices $u\in S$ and/or $v\in S$, and thus $u$ or $v$ may have been included in $T$. But the vertices $u,v$ do not help distinguish between paths $P_1$ and $P_2$. However, observe that Algorithm~\ref{alg:approximate} also includes all neighbors of the vertices in $S$ into $T$. Further each vertex in $V(C)$ has at least two of its neighbors in $V(C)$. Thus at least one vertex in $V(C)$, other than $u$ and $v$, will have been necessarily included in $T$. This violates the claim that no vertex other than $u$ or $v$ belongs to $T$, contradicting the assumption that $T$ is not a tracking set for $G$.

Next we explain the approximation ratio $2(\delta+1)$. From~\cite{tr-j}, it is known that each tracking set is also a feedback vertex set. Hence, for a graph $G$, the size of a minimum FVS serves as a lower bound for the size of an optimum tracking set for $G$. Thus when Algorithm~\ref{alg:approximate} includes a $2$-approximate FVS $S$, into the tracking set $T$, the size of $T$ is at most $2\cdot OPT$, where $OPT$ is the size of an optimum tracking set for $G$. Further, for each vertex in $S$, all of its neighbors are also included into the tracking set $T$. Since the maximum degree of $G$ is upper bounded by $\delta$, for each vertex in $S$, additional $\delta$ vertices are included in $T$. Thus the size of $T$ is at most $2(\delta+1)\cdot OPT$.
\qed
\end{proof}

Next we prove that Algorithm~\ref{alg:approximate} runs in polynomial time.

\begin{lemma}
\label{lemma:approximate-algo-time}
Algorithm~\ref{alg:approximate} runs in time $\Oh(n^2)$.
\end{lemma}
\begin{proof}
Algorithm~\ref{alg:approximate} starts by applying Reduction Rule~\ref{red:stpath-undirected} that can be applied in quadratic time. Next we find a $2$-approximate feedback vertex set $S$ for the input graph, using the algorithm given in~\cite{twofvs} in $\Oh(min\{|E|\log|V|,|V|^2\})$ time. We include $S$ in tracking set $T$. Next, for each vertex $v\in S$, we add $N(v)$ to $T$. This step takes $\Oh(n^2)$ time. Thus the overall time taken is $\Oh(n^2)+\Oh(min\{|E|\log|V|,|V|^2\})+\Oh(n^2)$. Hence the algorithm runs in total $\Oh(n^2)$ time.\qed
\end{proof}

From Lemma~\ref{lemma:approximate-algo-correctness} and Lemma~\ref{lemma:approximate-algo-time}, we have the following theorem.

\begin{theorem}
\label{theorem:approximate}
For an undirected graph $G$ on $n$ vertices such that the maximum degree of vertices in $G$ is $\delta$, there exists an $\Oh(n^2)$ algorithm that finds a $2(\delta+1)$-approximate tracking set for $G$.
\end{theorem}

The approximation ratio for our algorithm can be improved slightly by using the improved approximation bounds known for FVS in bounded degree graphs~\cite{bafna-fvs}.

\section{Reconstructing Paths using Trackers}
\label{sec:path-recon}
In real-world applications, it might be required to identify the \stpath which corresponds to a given sequence of trackers. Banik et al.~\cite{ciac17} gave a polynomial time algorithm to reconstruct the \sstpath corresponding to a subset of trackers, given a tracking set for \sstpaths. Here we give an algorithm that works for all \stpaths.
Given a graph $G$ and a sequence of trackers $\pi$, such that $V(\pi)\subseteq T$, where $T$ is a tracking set $T$ of constant size for $G$, the algorithm identifies the unique \stpath in $G$ that corresponds to $\pi$. Our algorithm works for both undirected graphs as well as tournaments.

\begin{lemma}
\label{lemma:path-reconstruction}
Given a graph $G$, and tracking set $T$ for $G$, and a sequence of trackers $\pi$, the unique \stpath corresponding to $\pi$ can be found in polynomial time.
\end{lemma}
\begin{proof}
Let $V(\pi)$ denote the vertices in the sequence $\pi$ and let $|V(\pi)|=k$. Let $P$ be the path that we need to find, i.e. the unique \stpath in $G$ that corresponds to $\pi$.
Let $\pi=(s,v_1,v_2,\dots,t)$ be the sequence of trackers received as part of the input. Let $S$ be the set of pairs of vertices formed from consecutive vertices in $\pi$, i.e. $S=\{ \{s,v_1\},\{v_1,v_2\},\dots,\{v_k,t\} \}$. Since $\pi$ corresponds to $P$, $V(P)$ should not contain any trackers from $T$, other than those in $\pi$. Now we need to find the path that passes through the sequence of vertices in $\pi$. In order to do so we find the vertex disjoint paths between $v_i$ and $v_{i+1}$, where $v_0=s$ and $v_{k+1}=t$. The sub paths between the pairs of vertices in $\pi$ should be vertex disjoint. We create a copy $v_i'$ for each vertex $v_i$ in $\pi$, and introduce and edge between $v_i'$ and each vertex in $N(v_i)$ in the graph $G$. We create a new set $S'=\{ \{s,v_1\},\{v_1',v_2\},\{v_2',v_3\}\dots,\{v_{k-1},v_k\},\{v_k',t\} \}$ and $V(S')$ be the set of all vertices in $S'$. Consider the graph $G'=G- (T\setminus V(S'))$. If $G$ is an undirected graph, then using the algorithm for disjoint paths in undirected graphs from~\cite{KAWARABAYASHI2012424}, find the vertex disjoint paths between the pairs of vertices in $S'$, in the graph $G'$. If $G$ is a tournament graphs, then using the algorithm for disjoint paths in tournaments from~\cite{chudnovsky}, find the vertex disjoint paths between the pairs of vertices in $S'$, in the graph $G'$. Since disjoint path problem can be solved in polynomial time for undirected graphs and tournaments~\cite{KAWARABAYASHI2012424},\cite{chudnovsky}, we can perform this step in polynomial time. Observe that the sequence of these vertex disjoint paths will form an \stpath in $G'$, which will also be an \stpath in $G$. Next we prove that the path found will be a unique \stpath. Suppose not. Then there exists two \stpaths in $G$, that contain the sequence of trackers $\pi$. However, this contradicts the assumption that $T$ is a tracking set for $G$. Observe that if the paths between pairs of vertices $v_i$ and $v_{i+1}$ are not vertex disjoint, this results in violation of tracking set condition, as there are two vertex disjoint paths between a pair of vertices that have disjoint paths to $s$ and $t$ themselves. This contradicts the assumption that $T$ is a tracking set for $G$.
\qed
\end{proof}

\section{Tracking Edge Set for Undirected Graphs}

In this section we study the problem of identifying \stpaths in an undirected edge weighted graph using the edges of the graph. For a graph $G$, we define a \textit{tracking edge set} as the set of edges whose intersection with each \stpath results in a unique sequence of edges. Here we allow parallel edges in the input graph. We formally define the problem of tracking paths using edges as follows.

\defproblem{\tpe $(G,s,t)$}{An undirected edge weighted graph $G=(V,E)$ with terminal vertices $s$ and $t$.}
{Find a minimum weight tracking edge set  $T\subseteq E$ for $G$.}
\medskip

We start by first applying Reduction Rule~\ref{red:stpath-undirected}, which ensures that each vertex and edge in the graph participates in some \stpath. Next we prove that each cycle in the reduced graph needs an edge as a tracker.

\begin{lemma}
\label{lemma:each-cycle-needs-an-edge}
For a reduced graph $G=(V,E)$, if $T\subseteq E$ is a tracking edge set, then each cycle in $G$ contains an edge $e$ such that $e\in T$.
\end{lemma}

\begin{figure}[ht]
\centering
\includegraphics[scale=0.5]{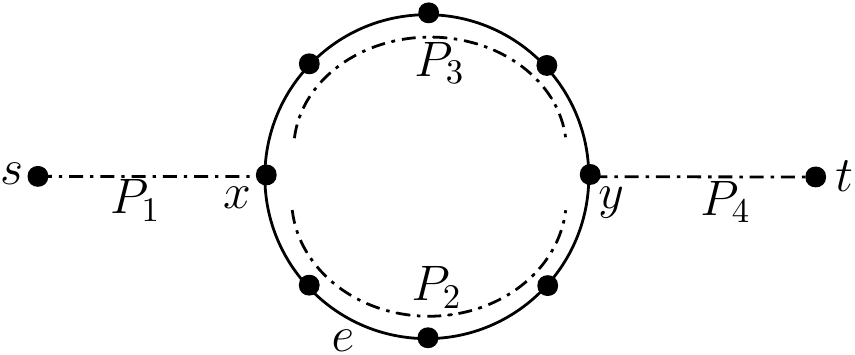} 
\caption{Cycle without any tracking edges} 
\label{fig:cycle-no-trackers}
\end{figure}

\begin{proof}
Suppose the claim does not hold. Then there exists a cycle $C$ in graph $G$, such that $E(C)\cap T=\emptyset$, i.e. none of the edges in $C$ belong to $T$. Consider an edge $e\in V(C)$. Due to Reduction Rule~\ref{red:stpath-undirected}, $e$ participates in an \stpath, say $P$. Let $u$ be the first vertex of $C$ that appears in path $P$ while traversing from $s$ to $t$. Similarly, let $v$ be the last vertex of $C$ that appears in path $P$ while traversing from $s$ to $t$. See Figure~\ref{fig:cycle-no-trackers}. Observe that $u$ and $v$ serve as local source and sink respectively for the cycle $C$, and there exist exactly two vertex disjoint paths between $u$ and $v$ in $C$. Since none of the edges in $C$ are part of the tracking edge set $T$, this leads to two \stpaths in $G$ with exactly same sequence of edges. This contradicts the fact that $T$ is a tracking edge set for $G$.
\qed
\end{proof}

Next we prove that a feedback edge set (FES) is a tracking edge set for a reduced graph. An FES is a set of edges whose removal makes the graph acyclic.

\begin{lemma}
\label{lemma:fas-is-tes}
For a reduced graph $G$, a feedback edge set $F$ is also a tracking edge set for $G$.
\end{lemma}
\begin{proof}
Consider graph $G=(V,E)$ reduced by Reduction Rule~\ref{red:stpath-undirected}, and an FES $F\subseteq E$ for $G$. We claim that $T=F$ is a tracking edge set for $G$. Suppose not. Then there exists two \stpaths, say $P_1$ and $P_2$, in $G$, such that the sequence of tracking edges in both these paths is the same. The graph induced by $P_1$ and $P_2$ contains at least one cycle, say $C$. See Figure~\ref{fig:paths-cycle}. Since $P_1$ and $P_2$ contain the same sequence of tracking edges, there must be no edge in cycle $C$ that belongs to $T$. However, since $T$ is an FES for $G$, it must necessarily contain an edge, say $e$,  from the cycle $C$ marked as a tracking edge. Observe that $e$ can belong to either $P_1$ and $P_2$, but not both of them. This contradicts the assumption that $P_1$ and $P_2$ contain the same sequence of tracking edges.
\qed
\end{proof}

Although finding a minimum FVS is a \textsc{NP}-hard problem, an FES can be found in polynomial time. We  now  prove that \tpe can be solved in polynomial time.

\begin{theorem}
\label{theorem:tracking-edge-set}
For an undirected edge-weighted graph $G$ on $n$ vertices, \tpe can be solved in $\Oh(n^2)$ time.
\end{theorem}
\begin{proof}
Let $G$ be an undirected edge-weighted graph on $n$ vertices.
From Lemma~\ref{lemma:fas-is-tes} it is known that an FES is a tracking edge set for $G$. In order to find a minimum weighted tracking edge set for $G$, we first find a maximum weight spanning tree $T$ for $G$ using Prim's algorithm or Kruskal's algorithm in $\Oh(n^2)$ time. Now the edges in $G-T$ comprise of a minimum weight FES, which is also a minimum weight tracking edge set for $G$.
\qed
\end{proof}

A path reconstruction algorithm similar to the one mentioned in Section~\ref{sec:path-recon} can be given by considering a sequence of tracking edges, and finding vertex disjoint paths between their endpoints in the graph obtained after removal of remaining tracking edges from the tracking edge set for that graph.

\section{Conclusions}
\label{sec:concl}

In this paper, we give polynomial time results for some variants of the \tp problem. Specifically, we solve \tp for chordal graphs and tournaments, along with giving an approximation algorithm for degree bounded graphs. We also analyze the problem \tpe, and prove it to be polynomial time solvable. A constructive algorithm has also been given that helps identify an \stpath, given the unique sequence of trackers it contains. Future scope of this work lies in improving the running times of these algorithms and identifying more graph classes where \tp may be easily solvable. Open problems include finding approximation algorithms for the \textsc{NP}-hard variants (other than bounded degree version) of the problem for both undirected and directed graphs.

\subsection*{Acknowledgement} We thank Prof. Venkatesh Raman for the insightful discussions and suggestions.

%
%
%
\bibliographystyle{splncs04}
\bibliography{tracking}

\end{document}